\documentclass[twocolumn,
showpacs,preprintnumbers,amsmath,amssymb]{revtex4}

\usepackage{amssymb,amsmath,amsthm}
\usepackage{graphicx}
\newtheorem{Thm}{Theorem}

\theoremstyle{definition}

\newcommand{\bra}[1]{{\left\langle #1 \right|}}
\newcommand{\ket}[1]{{\left| #1 \right\rangle}}

\newcommand{\T}{\mbox{$\mathrm{tr}$}}

\begin{document}
\title{Tsallis entropy and entanglement constraints in multi-qubit systems}

\author{Jeong San Kim}
\email{jekim@ucalgary.ca} \affiliation{
 Institute for Quantum Information Science,
 University of Calgary, Alberta T2N 1N4, Canada
}

\date{\today}

\begin{abstract}
We show that the restricted sharability and distribution of
multi-qubit entanglement can be characterized by Tsallis-$q$
entropy. We first provide a class of bipartite entanglement measures
named Tsallis-$q$ entanglement, and provide its analytic formula in
two-qubit systems for $1 \leq q \leq 4$. For $2 \leq q \leq 3$, we
show a monogamy inequality of multi-qubit entanglement in terms of
Tsallis-$q$ entanglement, and we also provide a polygamy inequality
using Tsallis-$q$ entropy for $1 \leq q \leq 2$ and $3 \leq q \leq
4$.
\end{abstract}

\maketitle

\section{Introduction}

Whereas classical correlations can be freely shared among parties in
multi-party systems, quantum correlation especially quantum
entanglement is known to have some restriction in its sharability
and distribution. For example, in a tripartite system consisting of
parties $A$, $B$ and $C$, let us assume $A$ is maximally entangled
with both $B$ and $C$ simultaneously. Because maximal entanglement
can be used to teleport an arbitrary unknown quantum
state~\cite{tele}, $A$ can teleport an unknown state $\rho$ to $B$
and $C$ by using the simultaneous maximal entanglement. Now, each
$B$ and $C$ has an identical copy of $\rho$, and this means cloning
an unknown state $\rho$, which is impossible by {\em no-cloning}
theorem~\cite{noclon}. In other words, the assumption of
simultaneous maximal entanglement of $A$ with $B$ and $C$ is quantum
mechanically forbidden.

This restricted sharability of quantum entanglement is known as the
{\em Monogamy of Entanglement}~(MoE)~\cite{T04}, and it was also
shown to play an important role in many applications of quantum
information processing. For instance, in quantum cryptography, MoE
can be used to restrict the possible correlation between authorized
users and the eavesdropper, which is the basic concept of the
security proof~\cite{m}.

For three-qubit systems, MoE was first characterized in forms of a
mathematical inequality using {\em concurrence}~\cite{ww} as the
bipartite entanglement measure. This characterization is known as
{\em CKW inequality} named after its establishers, Coffman, Kundu
and Wootters~\cite{ckw}, and it was also generalized for multi-qubit
systems later~\cite{ov}.

MoE in multi-qubit systems is mathematically well-characterized in
terms of concurrence, it is however also known that CKW-type
characterization for MoE is not generally true for other
entanglement measures such as {\em Entanglement of
Formation}~(EoF)~\cite{bdsw}: Even in multi-qubit systems, there
exists an counterexample that violates CKW-type inequality in terms
of EoF.

As bipartite entanglement measures, both concurrence and EoF of a
bipartite pure state $\ket{\psi}_{AB}$ quantify the uncertainty of
the subsystem $\rho_{A}=\T_{B} \ket{\psi}_{AB}\bra{\psi}$. For the
case when $\ket{\psi}_{AB}$ is a two-qubit state, the uncertainty of
$\rho_A$ is completely determined by a single parameter.
Furthermore, the extension of concurrence and that of Eof for a
mixed state $\rho_{AB}$ are based on the same method of {\em
convex-roof extension}, which minimizes the average entanglement
over all possible pure state decompositions of $\rho_{AB}$. In other
words, concurrence and EoF for two-qubit states are essentially
equivalent based on the same concept, the uncertainty of the
subsystem. Moreover, it was also shown that these two measures are
related by an monotone-increasing convex function~\cite{ww}.

However, these two equivalent measures for two-qubit systems show
very different properties in multipartite systems in characterizing
MoE, and this exposes the importance of having proper entanglement
measures to characterize MoE even in multi-qubit systems. Moreover,
for the study of general MoE in multipartite higher-dimensional
quantum systems, having a proper bipartite entanglement measure is
one of the most important and necessary things that must precede.

As generalizations of von Neumann entropy, there are two
representative classes of entropies quantifying the uncertainty of
quantum systems: One is quantum R\'enyi entropy~\cite{renyi, horo},
and the other is quantum Tsallis entropy~\cite{tsallis, lv}. Both of
them are one-parameter classes parameterized by a nonnegative real
number $q$, having von Neumann entropy as a special case when
$q\rightarrow 1$. Recently, it was shown that R\'enyi entropy can be
used for CKW-type characterization of multi-qubit monogamy
~\cite{ks2}.

Here, we show that Tsallis entropy can characterize MoE in
multi-qubit systems for a selective choice of the parameter $q$.
Using quantum Tsallis entropy of order $q$ (or Tsallis-$q$ entropy),
we first provide an one-parameter class of bipartite entanglement
measures, {\em Tsallis-$q$ entanglement}, and provide its analytic
formula for arbitrary two-qubit states when $1\leq q\leq 4$. This
class contains EoF as a special case when $q\rightarrow1$.
Furthermore, we show the monogamy inequality of multi-qubit systems
in terms of Tsallis-$q$ entanglement for $2\leq q\leq3$. For $1\leq
q \leq 2$ or $3\leq q \leq 4$, we also provide a polygamy inequality
of multi-qubit entanglement using Tsallis-$q$ entropy.

This paper is organized as follows. In Section~\ref{Subsec:
definition}, we recall the definition of Tsallis-$q$ entropy, and
define Tsallis-$q$ entanglement and its dual quantity for bipartite
quantum states. In Section~\ref{Subsec: 2formula}, we provide an
analytic formula of Tsallis-$q$ entanglement for arbitrary two-qubit
states when $1\leq q \leq 4$. In Section~\ref{Sec: monopoly}, we
derive a monogamy inequality of multi-qubit entanglement in terms of
Tsallis-$q$ entanglement for $2\leq q\leq3$. We also provide a
polygamy inequality of multi-qubit entanglement for $1\leq q \leq 2$
or $3\leq q \leq 4$. Finally, we summarize our results in
Section~\ref{Conclusion}.


\section{Tsallis-$q$ Entanglement}
\label{Sec: Tqentanglement}
\subsection{Definition}
\label{Subsec: definition}
For any quantum state $\rho$, its Tsallis-$q$ entropy is defined as
\begin{equation}
T_{q}(\rho)=\frac{1}{q-1}\left(1- \T \rho^{q}\right),
\label{r-entropy}
\end{equation}
for any $q >0$ and $q \neq 1$. For the case when $\alpha$ tends to
1, $T_{q}(\rho)$ converges to the von Neumann entropy, that is
\begin{align}
\lim_{q \rightarrow 1}T_{q}(\rho)=-\T \rho\log\rho 
=S(\rho). \label{T1}
\end{align}
In other words, Tsallis-$q$ entropy has a singularity at $q=1$, and
it can be replaced by von Neumann entropy. Throughout this paper, we
will just consider $T_1(\rho)=S(\rho)$ for any quantum state $\rho$.

For a bipartite pure state $\ket{\psi}_{AB}$ and each $q>0$,
Tsallis-$q$ entanglement is
\begin{equation}
{\mathcal T}_{q}\left(\ket{\psi}_{AB} \right):=T_{q}(\rho_A),
\label{TEpure}
\end{equation}
where $\rho_A=\T _{B} \ket{\psi}_{AB}\bra{\psi}$ is the reduced
density matrix onto subsystem $A$. For a mixed state $\rho_{AB}$, we
define its  Tsallis-$q$ entanglement via convex-roof extension, that
is,
\begin{equation}
{\mathcal T}_{q}\left(\rho_{AB} \right):=\min \sum_i p_i {\mathcal T}_{q}(\ket{\psi_i}_{AB}),
\label{TEmixed}
\end{equation}
where the minimum is taken over all possible pure state
decompositions of $\rho_{AB}=\sum_{i}p_i
\ket{\psi_i}_{AB}\bra{\psi_i}$.

As a dual quantity to Tsallis-$q$ entanglement, we also define
{\em Tsallis-$q$ entanglement of Assistance} (TEoA) as
\begin{equation}
{\mathcal T}^a_{q}\left(\rho_{AB} \right):=\max \sum_i p_i {\mathcal T}_{q}(\ket{\psi_i}_{AB}),
\label{TEoA}
\end{equation}
where the maximum is taken over all possible pure state
decompositions of $\rho_{AB}$.

Because Tsallis-$q$ entropy converges to von Neumann entropy when $q$ tends to 1,
we have
\begin{align}
\lim_{q\rightarrow1}{\mathcal T}_{q}\left(\rho_{AB} \right)=E_{\rm f}\left(\rho_{AB} \right),
\end{align}
where $E_{\rm f}(\rho_{AB})$ is the EoF of $\rho_{AB}$ defined
as~\cite{bdsw}
\begin{equation}
E_{\rm f}(\rho_{AB})=\min \sum_{i}p_i S(\rho^{i}_{A}). \label{eof}
\end{equation}
Here, the minimization is taken over all possible pure state
decompositions of $\rho_{AB}$, such that,
\begin{equation}
\rho_{AB}=\sum_{i} p_i |\phi^i\rangle_{AB}\langle\phi^i|,
\label{decomp}
\end{equation}
with $\T_{B}|\phi^i\rangle_{AB}\langle\phi^i|=\rho^{i}_{A}$.
In other words, Tsallis-$q$ entanglement is one-parameter generalization of EoF, and
the singularity of ${\mathcal T}_{q}\left(\rho_{AB}\right)$ at $q=1$ can be replaced by $E_{\rm f}(\rho_{AB})$.

Similarly, we have
\begin{align}
\lim_{q\rightarrow1}{\mathcal T}^a_{q}\left(\rho_{AB}
\right)=E^a\left(\rho_{AB} \right),
\label{TsallistoEoA}
\end{align}
where $E^a(\rho_{AB})$ is the {\em Entanglement of Assistance}~(EoA)
of $\rho_{AB}$ defined as~\cite{cohen}
\begin{equation}
E^a(\rho_{AB})=\max \sum_{i}p_i S(\rho^{i}_{A}). \label{eoa}
\end{equation}
Here, the maximum is taken over all possible pure state
decompositions of $\rho_{AB}$, such that,
\begin{equation}
\rho_{AB}=\sum_{i} p_i |\phi^i\rangle_{AB}\langle\phi^i|,
\label{decomp2}
\end{equation}
with $\T_{B}|\phi^i\rangle_{AB}\langle\phi^i|=\rho^{i}_{A}$.

\subsection{Analytic formula for two-qubit states}
\label{Subsec: 2formula}
Before we provide an analytic formula for Tsallis-$q$ entanglement in two-qubit systems,
let us first recall the definition of concurrence and its
functional relation with EoF in two-qubit systems.

For any bipartite pure state $\ket \psi_{AB}$, its concurrence,
$\mathcal{C}(\ket \psi_{AB})$ is defined as~\cite{ww}
\begin{equation}
\mathcal{C}(\ket \psi_{AB})=\sqrt{2(1-\T\rho^2_A)}, \label{pure
state concurrence}
\end{equation}
where $\rho_A=\T_B(\ket \psi_{AB}\bra \psi)$. For a mixed state
$\rho_{AB}$, its concurrence is defined as
\begin{equation}
\mathcal{C}(\rho_{AB})=\min \sum_k p_k \mathcal{C}({\ket
{\psi_k}}_{AB}). \label{mixed state concurrence}
\end{equation}
where the minimum is taken over all possible pure state
decompositions, $\rho_{AB}=\sum_kp_k{\ket {\psi_k}}_{AB}\bra
{\psi_k}$.

For two-qubit systems, concurrence is known to have an analytic
formula~\cite{ww}; for any two-qubit state $\rho_{AB}$,
\begin{equation}
\mathcal{C}(\rho_{AB})=\max\{0, \lambda_1-\lambda_2-\lambda_3-\lambda_4\},
\label{C_formula}
\end{equation}
where $\lambda_i$'s are the eigenvalues, in decreasing order, of
$\sqrt{\sqrt{\rho_{AB}}\tilde{\rho}_{AB}\sqrt{\rho_{AB}}}$ and
$\tilde{\rho}_{AB}=\sigma_y \otimes\sigma_y
\rho^*_{AB}\sigma_y\otimes\sigma_y$ with the Pauli operator
$\sigma_y$. Furthermore, the relation between concurrence and EoF of
a two-qubit mixed state $\rho_{AB}$ (or a pure state
$\ket{\psi}_{AB} \in \mathbb{C}^2 \otimes \mathbb{C}^{d}$,
$d\geq2$), can be given as a monotone increasing, convex
function~\cite{ww}, such that
\begin{equation}
 E_{\rm f} (\rho_{AB}) = {\mathcal E}(\mathcal{C}\left(\rho_{AB}\right)),
\end{equation}
where
\begin{equation} {\mathcal E}(x) = H\Bigl({1\over 2} + {1\over
2}\sqrt{1-x^2}\Bigr), \hspace{0.5cm}\mbox{for } 0 \le x \le 1,
\label{eps}
\end{equation}
with the binary entropy function $H(t) = -[t\log t + (1-t)\log
(1-t)]$. In other words, the analytic formula of concurrence as well
as its functional relation with EoF lead us to an analytic formula
for EoF in two-qubit systems.

For any $2\otimes d$ pure state $\ket{\psi}_{AB}$ (especially a
two-qubit pure state) with its Schmidt decomposition
$\ket{\psi}_{AB}=\sqrt{\lambda_{0}}\ket{0
0}_{AB}+\sqrt{\lambda_{1}}\ket{11}_{AB}$, its Tsallis-$q$
entanglement is
\begin{align}
{\mathcal T}_{q}\left(\ket{\psi}_{AB} \right)=T_{q}(\rho_A)
=\frac{1}{q-1}\left(1-\lambda_0^{q}-\lambda_1^{q} \right).
\label{TE2pure}
\end{align}

Because the concurrence of $\ket{\psi}_{AB}$ is
\begin{align}
\mathcal{C}(\ket \psi_{AB})=\sqrt{2(1-\T\rho^2_A)}
=\sqrt{\lambda_0 \lambda_1},
\end{align}
 it can be easily verified that
\begin{equation}
{\mathcal T}_{q}\left(\ket{\psi}_{AB}
\right)=g_{q}\left(\mathcal{C}(\ket \psi_{AB}) \right),
\label{relationpure}
\end{equation}
where $g_{q}(x)$ is an analytic function defined as
\begin{align}
g_{q}(x):=&\frac{1}{q-1}\left[1-\left(\frac{1+\sqrt{1-x^2}}{2}\right)^{q}
-\left(\frac{1-\sqrt{1-x^2}}{2}\right)^{q}\right]
\label{g_q}
\end{align}
on $0 \leq x \leq 1$. In other words, for any $2\otimes d$ pure
state $\ket{\psi}_{AB}$, we have a functional relation between its
concurrence and Tsallis-$q$ entanglement
 for each $q >0$. Note that $g_{q}(x)$ converges to the function ${\mathcal E}(x)$
in Eq.~(\ref{eps}) for the case when
$q$ tends to 1.

It was shown that there exists an optimal decomposition for the
concurrence of a two-qubit mixed state such that every pure state
concurrence in the decomposition has the same value~\cite{ww}: For
any two-qubit state $\rho_{AB}$, there exists a pure state
decomposition $\rho_{AB}=\sum_{i}p_i \ket{\phi_i}_{AB}\bra{\phi_i}$
such that
\begin{equation}
\mathcal{C}(\rho_{AB})=\sum_i p_i \mathcal{C}({\ket {\phi_i}}_{AB}),
\label{Copt}
\end{equation}
and
\begin{equation}
\mathcal{C}(\ket {\phi_i}_{AB})=\mathcal{C}(\rho_{AB}),
\label{Cphii}
\end{equation}
for each $i$. Based on this, one possible sufficient condition for
the relation in Eq.~(\ref{relationpure}) to be also true for
two-qubit mixed states is that the function $g_q(x)$ is
monotonically increasing and convex~\cite{suffi}. In other words, we
have
\begin{equation}
{\mathcal T}_{q}\left(\rho_{AB} \right)=g_{q}\left(\mathcal{C}(\rho_{AB}) \right)
\label{relationmixed}
\end{equation}
for any two-qubit mixed state $\rho_{AB}$ provided that $g_q(x)$ is
monotonically increasing and convex. Moreover, for the range of $q$
where $g_q(x)$ is monotonically increasing and convex,
Eq.~(\ref{relationmixed}) also implies an analytic formula of
Tsallis-$q$ entanglement for any two-qubit state.

Now, let us consider the monotonicity and convexity of $g_q(x)$ in
Eq.~(\ref{g_q}). Because $g_q(x)$ is an analytic function on $0\leq
x\leq1$, its monotonicity and convexity follow from the
nonnegativity of its first and second derivatives.

By taking the
first derivative of $g_q(x)$, we have
\begin{equation}
\frac{{\rm d}g_q(x)}{{\rm d}x}=\frac{qx\left[ {\left(1+\sqrt{1-x^2}\right)}^{q-1}-
{\left(1-\sqrt{1-x^2}\right)}^{q-1} \right]}{2^q(q-1)\sqrt{1-x^2}},
\label{1deri}
\end{equation}
which is always nonnegative on $0\leq x\leq1$ for $q>0$. It is also direct
to check that Eq.~(\ref{1deri}) is strictly positive for $0<x<1$. In
other words, $g_q(x)$ is a strictly monotone-increasing function for
any $q>0$.

For the second derivative of $g_q(x)$, we have
\begin{widetext}
\begin{align}
\frac{{\rm d}^2g_q(x)}{{\rm d}x^2}
=&\alpha\left[\frac{{\left(1+\sqrt{1-x^2}\right)}^{q-2}}{1-x^2}
\left(\frac{1+\sqrt{1-x^2}}{\sqrt{1-x^2}}-x^2(q-1)\right)
-\frac{{\left(1-\sqrt{1-x^2}\right)}^{q-2}}{1-x^2}
\left(\frac{1-\sqrt{1-x^2}}{\sqrt{1-x^2}}+x^2(q-1)\right)\right]
\label{2deri}
\end{align}
\end{widetext}
where $\alpha=\frac{q}{2^q(q-1)}$. Here, we first prove that
$g_q(x)$ is not convex for $q\geq5$ by showing the existence of
$x_0$ between 0 and 1 such that $\frac{{\rm d}^2g_q(x_0)}{{\rm
d}x^2}$ is negative. To see this, first note that the second term of
the right-hand side in Eq.~(\ref{2deri}) is always negative for
$0<x<1$ if $q>1$. Thus, it suffices to show that the first term of
the right-hand side in Eq.~(\ref{2deri}) is nonpositive at $x_0 \in
(0,1)$ for $q\geq5$. Furthermore, the only factor of the first term
that can be negative is
\begin{equation}
\left(\frac{1+\sqrt{1-x^2}}{\sqrt{1-x^2}}-x^2(q-1)\right),
\label{1term}
\end{equation}
since both $\alpha$ and
$\frac{{\left(1+\sqrt{1-x^2}\right)}^{q-2}}{1-x^2}$ are always
positive at $x \in (0,1)$ if $q>1$. By defining a function such that
\begin{equation}
h(x)=\frac{1-\sqrt{1-x^2}}{x^2\sqrt{1-x^2}}+1,
\label{h}
\end{equation}
the nonpositivity of Eq.~(\ref{1term})
is equivalent to
\begin{equation}
q\geq h(x).
\label{1term2}
\end{equation}
Since $h(x)$ is an analytic function on $0<x<1$, it is direct to
verify that it has a critical point at $x_0=\frac{\sqrt{3}}{2}$ with
$g_q(x_0)=5$, which is the global minimum. In other words, for
$q\geq5$, there always exists $x_0 \in (0,1)$ making
Eq.~(\ref{1term}) nonpositive, and thus $g_{q}(x)$ is not convex for
this region of $q$.

For the region of $q<5$, let us first consider the function $g_q(x)$ of
the integer value $q$, that is $q=1,~2,~3$ and $4$. If $q\rightarrow1$,
$g_q(x)$ converges to ${\mathcal E}(x)$ in Eq.~(\ref{eps}), which is
already known to be convex on $0\leq x\leq1$. Furthermore, we have
\begin{equation}
g_2(x)=\frac{x^2}{2},~g_3(x)=\frac{3x^2}{8},~g_4(x)=\frac{8x^2-x^4}{24},
\label{gint}
\end{equation}
which are convex polynomials on $0\leq x\leq1$.

In fact, if we consider $\frac{d^2g_q(x)}{dx^2}$ in
Eq.~(\ref{2deri}) as a function of $x$ and $q$
\begin{equation}
l(x,q)=\frac{{\rm d}^2g_q(x)}{{\rm d}x^2},
\label{l}
\end{equation}
defined on the domain ${\mathcal D}=\{(x,q)|0\leq x \leq 1, 1\leq q
\leq 4 \}$, it is tedious but also straightforward to check that
$l(x,q)$ does not have any vanishing gradient in the interior of
$\mathcal D$, and its function value on the boundary of $\mathcal D$
is always nonnegative. Because $l(x,q)$ is analytic in the interior
of $\mathcal D$, and continuous on the boundary, $l(x,q)$ is
nonnegative through whole the domain $\mathcal D$, and this implies
the convexity of $g_q(x)$ for $1 \leq q \leq 4$. Thus, we have the
following theorem.
\begin{Thm}
For $1\leq q \leq 4$,
\begin{align}
g_{q}(x)=&\frac{1}{q-1}\left[1-\left(\frac{1+\sqrt{1-x^2}}{2}\right)^{q}
-\left(\frac{1-\sqrt{1-x^2}}{2}\right)^{q}\right]
\end{align}
is a monotonically-increasing convex function on $0\leq x \leq 1$.
Furthermore, for this range of $q$, any two-qubit state $\rho_{AB}$
has an analytic formula for its Tsallis-$q$ entanglement such that
${\mathcal T}_{q}\left(\rho_{AB}
\right)=g_{q}\left(\mathcal{C}(\rho_{AB}) \right)$ where
$\mathcal{C}(\rho_{AB})$ is the concurrence of $\rho_{AB}$.
\label{2formula}
\end{Thm}

Due to the continuity of $g_q(x)$ with respect to $q$, we can always
assure the convexity of $g_q(x)$ for some region of $q$ slightly
less than 1 or larger than 4. Furthermore, the continuity of
$l(x,q)$ in Eq.~(\ref{l}) also assures the existence of $q_0$
between 4 and 5, at which the convexity of $g_q(x)$ starts being
violated. However, it is generally hard to get an algebraic solution
of such $q_0$ since $\frac{{\rm d}^2g_q(x)}{{\rm d}x^2}$ in
Eq.~(\ref{2deri}) is not an algebraic function with respect to $q$.
Here, we have a numerical way of calculation to test various values
of $x$ and $q$, and it is illustrated in Figure~\ref{figderi2}.
\begin{figure}
\hfill
\parbox{4.3cm}{
\begin{center}
\includegraphics[width=\linewidth]{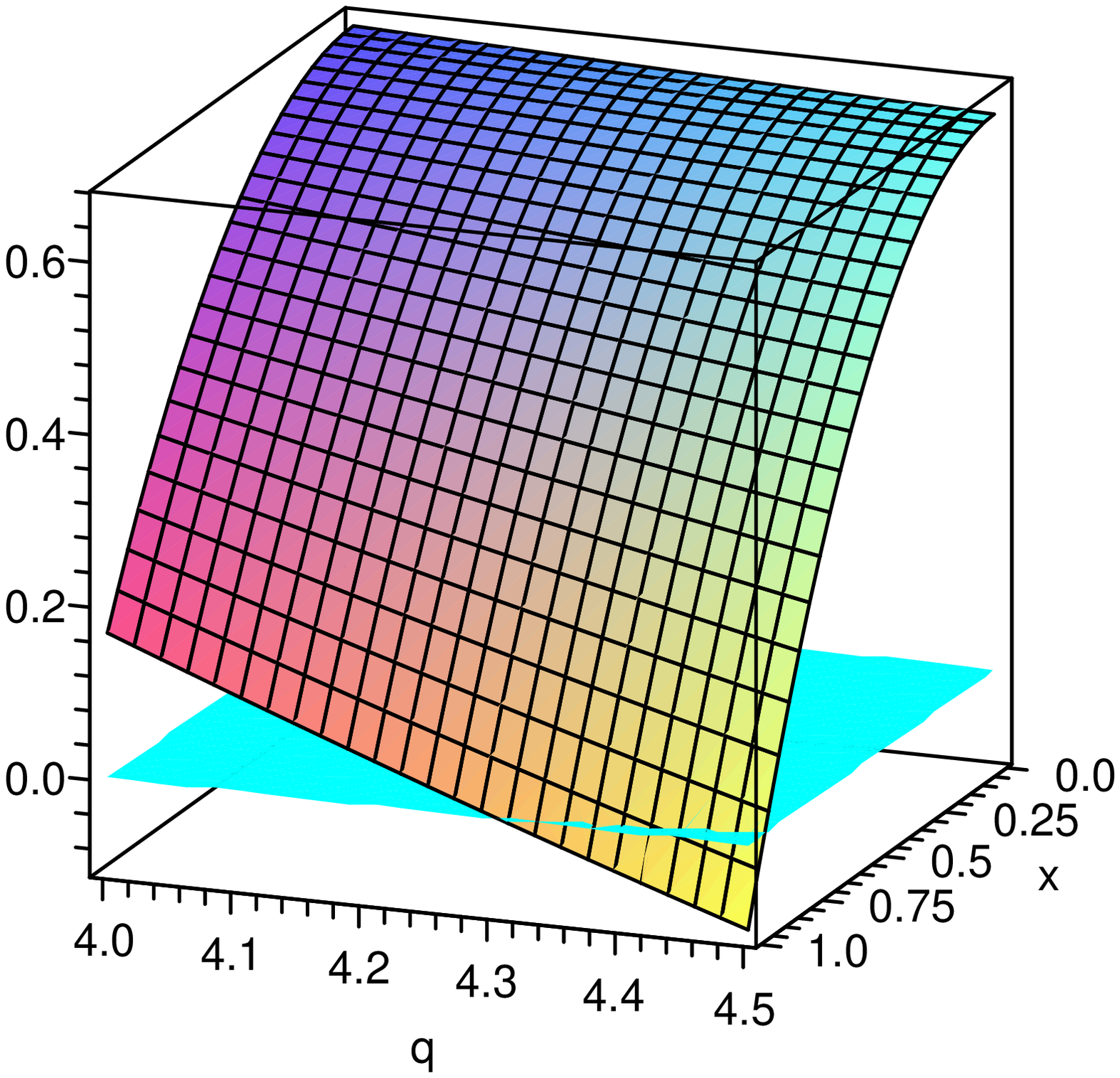}\\
(a)
\end{center}
}
\hfill
\parbox{4cm}{
\begin{center}
\includegraphics[width=\linewidth]{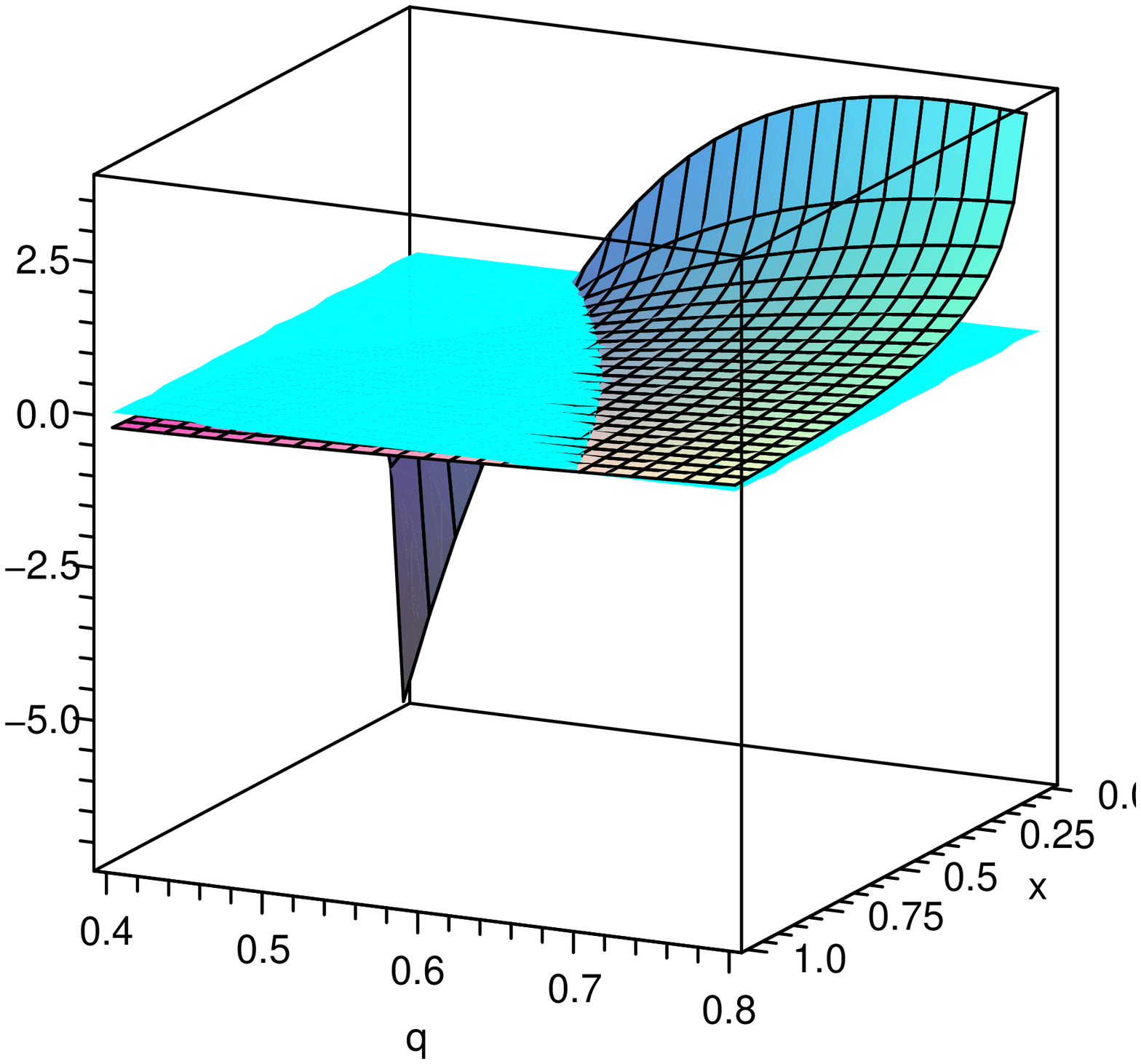}\\
(b)
\end{center}
}
\caption{The function values of $\frac{{\rm d}^2g_q(x)}{{\rm d}x^2}$ for $4\leq q\leq 4.5$ and $0.4\leq q\leq 0.8$
are illustrated in picture (a) and (b) respectively.
 }\label{figderi2}
\end{figure}
According to Figure~\ref{figderi2}, $g_q(x)$ is convex for the
region $0.7\leq q \leq 4.2$, and thus the analytic formula of
Tsallis-$q$ entanglement for two-qubit states in
Eq.~(\ref{relationmixed}) can also be claimed for this region of
$q$.


\section{Multi-qubit Entanglement constraint in terms of Tsallis-$q$ Entanglement}
\label{Sec: monopoly}

Using concurrence as the bipartite entanglement measure, the
monogamous property of a multi-qubit pure state
$\ket{\psi}_{A_1A_2\cdots A_n}$ was shown to have a mathematical
characterization as,
\begin{equation}
\mathcal{C}_{A_1 (A_2 \cdots A_n)}^2  \geq  \mathcal{C}_{A_1 A_2}^2
+\cdots+\mathcal{C}_{A_1 A_n}^2, \label{nCmono}
\end{equation}
where $\mathcal{C}_{A_1 (A_2 \cdots
A_n)}=\mathcal{C}(\ket{\psi}_{A_1(A_2\cdots A_n)})$ is the
concurrence of $\ket{\psi}_{A_1A_2\cdots A_n}$ with respect to the
bipartite cut between $A_1$ and the others, and
$\mathcal{C}_{A_1A_i}=\mathcal{C}(\rho_{A_1A_i})$ is the concurrence
of the reduced density matrix $\rho_{A_1A_i}$ for $i=2,\ldots,
n$~\cite{ckw,ov}.

As a dual value to concurrence, {\em Concurrence of Assistance} (CoA)~\cite{lve} of
a bipartite state $\rho_{AB}$ is defined as
\begin{equation}
\mathcal{C}^a(\rho_{AB})=\max \sum_k p_k \mathcal{C}({\ket {\psi_k}}_{AB}),
\label{CoA}
\end{equation}
where the maximum is taken over all possible pure state decompositions of
$\rho_{AB}=\sum_k p_k \ket{\psi_k}_{AB}\bra{\psi_k}$. Furthermore, it was also
shown that there exists a {\em polygamy} (or dual monogamy) relation of
multi-qubit entanglement in terms of CoA~\cite{gbs}:
For any multi-qubit pure state $\ket{\psi}_{A_1 \cdots A_n}$,
we have
\begin{equation} \mathcal{C}_{A_1 (A_2 \cdots
A_n)}^2  \leq  (\mathcal{C}^a_{A_1 A_2})^2
+\cdots+(\mathcal{C}^a_{A_1 A_n})^2, \label{nCdual}
\end{equation}
where $\mathcal{C}^a_{A_1 A_i}$ is the CoA of the reduced density
matrix $\rho_{A_1A_i}$ for $i=2,\ldots, n$.

Here, we show that this monogamous and polygamous property of
multi-qubit entanglement can also be characterized in terms of
Tsallis-$q$ entanglement and TEoA. Before this, we provide an
important property of the function $g_q(x)$ in Eq.~(\ref{g_q}) for
the proof of multi-qubit monogamy and polygamy relations.

For each $q>0$, let us define a two-variable function $m_q(x,y)$,
\begin{equation}
m_q(x,y):=g_q\left(\sqrt{x^2+y^2}\right)-g_q(x)-g_q(y),
\label{m_q}
\end{equation}
on the domain ${\mathcal D}=\{ (x,y)| 0\leq x, y, x^2+y^2 \leq1\}$.
Since $m_q(x,y)$ is continuous on the domain $\mathcal D$ and
analytic in the interior, its maximum or minimum values can arise
only at the critical points or on the boundary of $\mathcal D$. By
taking the first-order partial derivatives of $m_q(x,y)$, we have
its gradient
\begin{equation}
\nabla m_{p}(x, y)=\left(\frac{\partial
m_{p}(x,y)}{\partial x}, \frac{\partial
m_{p}(x,y)}{\partial y}\right)
\label{grad}
\end{equation}
where
\begin{widetext}
\begin{align}
\frac{\partial m_{q}(x,y)}{\partial x}=&\alpha x\left[\frac{{\left(1+\sqrt{1-x^2-y^2}\right)}^{q-1}-{\left(1+\sqrt{1-x^2-y^2}\right)}^{q-1}}{\sqrt{1-x^2-y^2}}
-\frac{{\left(1+\sqrt{1-x^2}\right)}^{q-1}-{\left(1+\sqrt{1-x^2}\right)}^{q-1}}{\sqrt{1-x^2}}\right]\nonumber\\
\frac{\partial m_{q}(x,y)}{\partial y}=&\alpha y\left[\frac{{\left(1+\sqrt{1-x^2-y^2}\right)}^{q-1}-{\left(1+\sqrt{1-x^2-y^2}\right)}^{q-1}}{\sqrt{1-x^2-y^2}}
-\frac{{\left(1+\sqrt{1-y^2}\right)}^{q-1}-{\left(1+\sqrt{1-y^2}\right)}^{q-1}}{\sqrt{1-y^2}}\right],
\label{2pderi}
\end{align}
\end{widetext}
with $\alpha=\frac{q}{2^q(q-1)}$.

Suppose there exists $(x_0, y_0)$ in the interior of $\mathcal D$
(that is, $0<x_0, y_0, x_0^2+y_0^2<1$) such that $\nabla m_{p}(x_0,
y_0)=0$. From Eq.~(\ref{2pderi}), it is straightforward to verify
that $\nabla m_{p}(x_0, y_0)=0$ is equivalent to
\begin{equation}
n_q(x_0)=n_q(y_0),
\label{x0y0}
\end{equation}
for an analytic function
\begin{equation}
n_q(t)=\frac{{\left(1+\sqrt{1-t^2}\right)}^{q-1}-{\left(1+\sqrt{1-t^2}\right)}^{q-1}}{\sqrt{1-t^2}},
\label{n_q}
\end{equation}
on $0<t<1$. Furthermore, it is straightforward to see that
$\frac{{\rm d}n_q(t)}{{\rm d}t}<0$ for $q>1$. In other words,
$n_q(t)$ is a strictly monotone-decreasing function with respect to
$t$ for $q>1$; therefore Eq.~(\ref{x0y0}) implies $x_0=y_0$.
However, from Eq.~(\ref{2pderi}), $\frac{\partial m_{q}(x_0,
y_0)}{\partial x}=0$ together with $x_0=y_0$ imply that
$n_q(\sqrt{2}x_0)=n_q(x_0)$, which contradicts to the strict
monotonicity of $n_q(t)$. Thus $m_q(x,y)$ has no vanishing gradient
in the interior of $\mathcal D$.

Now, let us consider the function values of $m_q(x,y)$ on the
boundary of $\mathcal D$. If $x=0$ or $y=0$, it is clear that
$m_q(x,y)=0$. For the case when $x^2+y^2=1$, $m_q(x,y)=0$ becomes a
single variable function
\begin{align}
b_q(x)=&\beta\left[\left(1+\sqrt{1-x^2}\right)^q +\left(1-\sqrt{1-x^2}\right)^q\right]\nonumber\\
&+\beta\left[\left(1+x\right)^q +\left(1-x\right)^q-2-2^q\right]
\label{m_qxy1}
\end{align}
with $\beta=\frac{1}{(q-1)2^q}$, which is an analytic function on
$0\leq x\leq1$. For the case when $q=2$ or $3$, it is clear form
Eq.~(\ref{gint}) that $m_q(x,y)=0$, and thus $b_q(x)=0$. If $q$ is
neither 2 nor 3, $b_q(x)$ has only one critical point at
$x=\frac{1}{\sqrt 2}$ for any $q>1$. Because $b_q(0)=b_q(1)=0$,
which are the function values at the boundary, the signs of the
function values of $b_q(x)$ are totally determined by that of
$b_q\left(\frac{1}{\sqrt 2} \right)$, which is the function value at
the critical point. Now, we have
\begin{align}
b_q\left(\frac{1}{\sqrt 2} \right)=&\frac{2}{(q-1)2^q}\left[\left(1+\frac{1}{\sqrt 2} \right)^q+
\left(1-\frac{1}{\sqrt 2} \right)^q\right]\nonumber\\
&-\frac{1}{(q-1)2^q}\left(2+2^{q}\right),
\label{bqsqrt1over2}
\end{align}
whose function value with respect to $q$ is illustrated in Figure~\ref{figbq}.

\begin{figure}
\includegraphics[width=4.5cm]{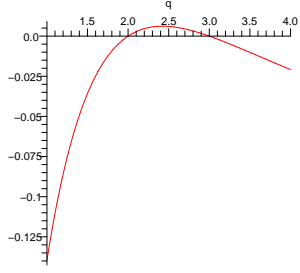}\\
\caption{The function values of $b_q\left(\frac{1}{\sqrt 2}\right)$ with respect to $q$
for $1<q\leq4$.}
\label{figbq}
\end{figure}

In other words, the function $m_q(x,y)$ in Eq.~(\ref{m_q}) has no
vanishing gradient in the domain $\mathcal D$ for $q>1$, and its
function values at the boundary of $\mathcal D$ is always
nonpositive for $1\leq q<2$ and $3<q\leq4$, whereas $m_q(x,y)$ is
always nonnegative for $2<q<3$. Thus, we have
\begin{equation}
g_q\left(\sqrt{x^2+y^2}\right)\leq g_q(x)+g_q(y)
\label{gqpoly}
\end{equation}
for $1<q<2$ and $3<q<4$, and
\begin{equation}
g_q\left(\sqrt{x^2+y^2}\right)\geq g_q(x)+g_q(y) \label{gqmono}
\end{equation}
for $2<q<3$. For the case when $q=2$ or 3, we have
\begin{equation}
g_q\left(\sqrt{x^2+y^2}\right)= g_q(x)+g_q(y).
\end{equation}

Now, we are ready to have the following theorem, which is the
monogamy inequality of multi-qubit entanglement in terms of
Tsallis-$q$ entanglement.

\begin{Thm}
For a multi-qubit state $\rho_{A_1 \cdots A_n}$ and $2\leq q \leq
3$, we have
\begin{equation}
{\mathcal T}_{q}\left( \rho_{A_1(A_2 \cdots A_n)}\right)\geq
{\mathcal T}_{q}(\rho_{A_1 A_2}) +\cdots+{\mathcal T}_{q}(\rho_{A_1
A_n}) \label{Tmono}
\end{equation}
where ${\mathcal T}_{q}\left( \rho_{A_1(A_2 \cdots A_n)}\right)$ is
the Tsallis-$q$ entanglement of $\rho_{A_1\left(A_2 \cdots
A_n\right)}$ with respect to the bipartite cut between $A_1$ and
$A_{2}\cdots A_{n}$, and ${\mathcal T}_{q}(\rho_{A_1 A_i})$ is the
Tsallis-$q$ entanglement of the reduced density matrix $\rho_{A_1
A_i}$ for $i=2,\cdots,n$. \label{Thm: mono}
\end{Thm}

\begin{proof}
For the case when $q=2$ or $3$, Eq.~(\ref{gint}) implies
\begin{equation}
{\mathcal T}_{2}\left(\rho_{AB}\right)=\frac{{{\mathcal
C}_{AB}}^2}{2},~{\mathcal
T}_{3}\left(\rho_{AB}\right)=\frac{3}{2}{{\mathcal C}_{AB}}^2,
\label{23}
\end{equation}
for any two-qubit mixed state or $2\otimes d$ pure state $\rho_{AB}$
and its concurrence ${\mathcal C}_{AB}$. Thus, the monogamy
inequality in Eq~(\ref{Tmono}) follows from Eqs.~(\ref{nCmono}) and
(\ref{23}).

For $2<q<3$, We first prove the theorem for $n$-qubit pure state
$\ket{\psi}_{A_1\cdots A_n}$. Note that Eq.~(\ref{nCmono}) is
equivalent to
\begin{equation}
\mathcal{C}_{A_1 (A_2 \cdots A_n)} \geq  \sqrt{\mathcal{C}_{A_1
A_2}^2 +\cdots+\mathcal{C}_{A_1 A_n}^2}, \label{nCmonoroot}
\end{equation}
for any $n$-qubit pure state $\ket{\psi}_{A_1 (A_2 \cdots A_n)}$.
Thus, from Eq.~(\ref{gqmono}) together with Eq.~(\ref{nCmonoroot}),
we have
\begin{align}
{\mathcal T}_{q}\left(\ket{\psi}_{A_1 (A_2 \cdots A_n)} \right)=&
g_{q}\left(\mathcal{C}_{A_1 (A_2 \cdots A_n)}\right)\nonumber\\
\geq&
g_{q}\left(\sqrt{\mathcal{C}_{A_1 A_2}^2 +\cdots+\mathcal{C}_{A_1 A_n}^2}\right)\nonumber\\
\geq&g_{q}\left(\mathcal{C}_{A_1 A_2}\right)\nonumber\\
&+g_{q}\left(\sqrt{\mathcal{C}_{A_1 A_3}^2
+\cdots+\mathcal{C}_{A_1 A_n}^2}\right)\nonumber\\
&~~~~~~~\vdots\nonumber\\
\geq& g_{q}\left(\mathcal{C}_{A_1 A_2}\right)+\cdots+g_{q}\left(\mathcal{C}_{A_1 A_n}\right)\nonumber\\
=& {\mathcal T}_{q}\left(\rho_{A_1A_2}\right)+\cdots +{\mathcal
T}_{q}\left(\rho_{A_1A_n}\right) \label{monoineq}
\end{align}
where the first equality is by the functional relation between the
concurrence and the Tsallis-$q$ entanglement for $2\otimes d$ pure
states, the first inequality is by the monotonicity of $g_{q}(x)$,
the other inequalities are by iterative use of Eq.~(\ref{gqmono}),
and the last equality is by Theorem~\ref{2formula}.

For a $n$-qubit mixed state  $\rho_{A_1(A_2\cdots A_n)}$, let
$\rho_{A_1(A_2\cdots A_n)}=\sum_j p_j \ket{\psi_j}_{A_1(A_2\cdots
A_n)}\bra{\psi_j}$ be an optimal decomposition such that ${\mathcal
T}_{q}\left(\rho_{A_1(A_2\cdots A_n)}\right)=\sum_j p_j {\mathcal
T}_{q}\left(\ket{\psi_j}_{A_1(A_2\cdots A_n)}\right)$.

Because each $\ket{\psi_j}_{A_1(A_2\cdots A_n)}$ in the
decomposition is an $n$-qubit pure state, we have
\begin{widetext}
\begin{align}
{\mathcal T}_{q}\left(\rho_{A_1(A_2\cdots A_n)}\right)=&\sum_j p_j
{\mathcal T}_{q}\left(\ket{\psi_j}_{A_1(A_2\cdots
A_n)}\right)\nonumber\\
\geq&\sum_j p_j\left({\mathcal
T}_{q}\left(\rho^j_{A_1A_2}\right)+\cdots +{\mathcal
T}_{q}\left(\rho^j_{A_1A_n}\right) \right)\nonumber\\
=&\sum_j p_j{\mathcal T}_{q}\left(\rho^j_{A_1A_2}\right)+\cdots
+\sum_j p_j{\mathcal
T}_{q}\left(\rho^j_{A_1A_n}\right) \nonumber\\
\geq&{\mathcal T}_{q}\left(\rho_{A_1A_2}\right)+\cdots +{\mathcal
T}_{q}\left(\rho_{A_1A_n}\right), \label{Tmonomixed}
\end{align}
\end{widetext}
where $\rho^j_{A_1A_i}$ is the reduced density matrix of
$\ket{\psi_j}_{A_1(A_2\cdots A_n)}$ onto subsystem $A_1A_i$ for each
$i=2,\cdots,n$ and the last inequality is by definition of
Tsallis-$q$ entanglement for each $\rho_{A_1A_i}$.
\end{proof}

Now, let us consider the polygamy of multi-qubit entanglement using
Tsallis-$q$ entropy. We first note that the function $g_{q}(x)$ in
Eq.~(\ref{g_q}) can also relate CoA and TEoA of a two-qubit state
$\rho_{AB}$: By letting $\rho_{AB}=\sum_{i} p_i
\ket{\psi_i}_{AB}\bra{\psi_i}$ be an optimal decomposition for its
CoA, that is,
\begin{equation}
\mathcal{C}^{a}\left( \rho_{AB} \right)=\sum_{i} p_i \mathcal {C}\left( \ket{\psi_i}_{AB}\right),
\label{CoAopt}
\end{equation}
we have
\begin{align}
g_{q}\left( \mathcal{C}^{a}\left( \rho_{AB} \right)\right)=&g_{q}
\left( \sum_{i} p_i \mathcal {C}\left( \ket{\psi_i}_{AB}\right)\right)\nonumber\\
\leq& \sum_{i} p_i  g_{q}\left(\mathcal {C}\left( \ket{\psi_i}_{AB}\right)\right)\nonumber\\
=& \sum_{i} p_i  {\mathcal T}_{q}\left( \ket{\psi_i}_{AB}\right)\nonumber\\
\leq&{\mathcal T}^a_{q}(\rho_{AB})
\label{gTEoA}
\end{align}
where the first inequality can be assured by the convexity of
$g_{q}(x)$ and the last inequality is by the definition of TEoA.
Because $g_{q}(x)$ is convex for $1\leq q \leq 4$, Eq.~(\ref{gTEoA})
is thus true for this region of $q$. Furthermore, $g_{q}(x)$
satisfies the property of Eq.~(\ref{gqpoly}) for $1\leq q \leq 2$ or
$3\leq q \leq 4$. Thus, we have the following theorem of the
polygamy inequality in multi-qubit systems.

\begin{Thm}
For any multi-qubit state $\rho_{A_1 \cdots A_n}$ and $1\leq q \leq
2$ or $3\leq q \leq 4$, we have
\begin{equation}
{\mathcal T}_{q}\left( \rho_{A_1(A_2 \cdots A_n)}\right)\leq
{\mathcal T}^a_{q}(\rho_{A_1 A_2}) +\cdots+{\mathcal
T}^a_{q}(\rho_{A_1 A_n}) \label{Tpoly}
\end{equation}
where ${\mathcal T}_{q}\left( \rho_{A_1(A_2 \cdots A_n)}\right)$ is
the Tsallis-$q$ entanglement of $\ket{\psi}_{A_1\left(A_2 \cdots
A_n\right)}$ with respect to the bipartite cut between $A_1$ and
$A_{2}\cdots A_{n}$, and ${\mathcal T}^a_{q}(\rho_{A_1 A_i})$ is the
TEoA of the reduced density matrix $\rho_{A_1 A_i}$ for
$i=2,\cdots,n$. \label{Thm: poly}
\end{Thm}

\begin{proof}
We first prove the theorem for a $n$-qubit pure state, and
generalize it into mixed states.

For the case when $q$ tends to 1, Tsallis-$q$ entanglement converges
to EoA in Eq.~(\ref{eoa}). It was shown that the polygamy inequality
of multi-qubit systems can be shown in terms of EoA~\cite{bgk}. For
the case when $q=2$ or 3, it is also straightforward from
Eqs.~(\ref{gint}) and (\ref{nCdual}).

For a $n$-qubit pure state $\ket{\psi}_{A_1(A_2\cdots A_n)}$ and
$1<q < 2$ or $3< q < 4$, let us first assume that $
(\mathcal{C}^a_{A_1 A_2})^2 +\cdots+(\mathcal{C}^a_{A_1 A_n})^2 \leq
1$ in Eq.~(\ref{nCdual}). Then we have
\begin{align}
{\mathcal T}_{q}\left(\ket{\psi}_{A_1(A_2 \cdots A_n)}\right)
&=g_{q}({\mathcal C}_{A_1(A_2 \cdots A_n)})\nonumber\\
&\leq g_{q}\left(\sqrt{(\mathcal{C}^a_{A_1 A_2})^2
+\cdots+(\mathcal{C}^a_{A_1 A_n})^2 }\right)\nonumber\\
&\leq g_{q}\left( \mathcal{C}^a_{A_1 A_2}\right)\nonumber\\
&~+g_{q}\left(\sqrt{(\mathcal{C}^a_{A_1 A_3})^2
+\cdots+(\mathcal{C}^a_{A_1 A_n})^2}\right)\nonumber\\
&~~~~~~~\vdots\nonumber\\
&\leq g_{q}\left( \mathcal{C}^a_{A_1 A_2}\right)+
+\cdots+g_{q}\left(
\mathcal{C}^a_{A_1 A_n}\right)\nonumber\\
&\leq {\mathcal T}^a_{q}\left(\rho_{A_1 A_2}\right)
+\cdots+{\mathcal T}^a_{q}\left(\rho_{A_1 A_n}\right),
\end{align}
where the first inequality is due to the monotonicity of the
function $g_{q}(x)$, the second and third inequalities are obtained
by iterative use of Eq.~(\ref{gqpoly}), and the last inequality is
by Eq.~(\ref{gTEoA}).

Now, let us assume that $ (\mathcal{C}^a_{A_1 A_2})^2
+\cdots+(\mathcal{C}^a_{A_1 A_n})^2 > 1$. Due to the monotonicity of
$g_q(x)$, we first note that
\begin{align}
{\mathcal T}_{q}\left(\ket{\psi}_{A_1(A_2 \cdots A_n)}\right)=&
g_{q}\left({\mathcal C}\left(\ket{\psi}_{A_1(A_2 \cdots A_n)}\right)\right)\nonumber\\
\leq& g_{q}\left(1\right)\nonumber\\
=&\frac{1}{q-1}\left(1-\frac{1}{2^{q-1}}\right)
\end{align}
for any multi-qubit pure state $\ket{\psi}_{A_1(A_2 \cdots A_n)}$,
and $q>1$. By letting $\gamma
=\frac{1}{q-1}\left(1-\frac{1}{2^{q-1}}\right)$, it is thus enough
to show that ${\mathcal T}^a_{q}(\rho_{A_1 A_2}) +\cdots+{\mathcal
T}^a_{q}(\rho_{A_1 A_n}) \geq \gamma$.

Here, we note that there exists $k \in \{2,\ldots ,n-1 \}$ such that
\begin{align}
&(\mathcal{C}^a_{A_1 A_2})^2 +\cdots+(\mathcal{C}^a_{A_1 A_k})^2
\leq 1,\nonumber\\
&(\mathcal{C}^a_{A_1 A_2})^2 +\cdots+(\mathcal{C}^a_{A_1 A_{k+1}})^2
>1.
\end{align}
If we let
\begin{equation} T:=(\mathcal{C}^a_{A_1 A_2})^2
+\cdots+(\mathcal{C}^a_{A_1 A_{k+1}})^2-1,
\end{equation}
we have
\begin{align}
\gamma =&g_{q}\left(1\right)\nonumber\\
=& g_{q} \left( \sqrt{(\mathcal{C}^a_{A_1 A_2})^2
+\cdots+(\mathcal{C}^a_{A_1
A_{k+1}})^2-T} \right)\nonumber\\
\leq& g_{q} \left( \sqrt{(\mathcal{C}^a_{A_1 A_2})^2
+\cdots+(\mathcal{C}^a_{A_1 A_k})^2} \right)\nonumber\\
&~~~~~~+g_{q} \left( \sqrt{(\mathcal{C}^a_{A_1 A_{k+1}})^2-T} \right)\nonumber\\
\leq& g_{q} \left( \mathcal{C}^a_{A_1 A_2}\right)+\cdots
+q_{q} \left( \mathcal{C}^a_{A_1 A_k}\right)+
q_{q} ( \mathcal{C}^a_{A_1 A_{k+1}})\nonumber\\
\leq& {\mathcal T}^a_{q}(\rho_{A_1 A_2})+\cdots + {\mathcal
T}^a_{q}(\rho_{A_1 A_n}), \label{nTmonopure2}
\end{align}
where the first inequality is by using Eq.~(\ref{gqpoly}) with
respect to $(\mathcal{C}^a_{A_1 A_2})^2+\cdots+(\mathcal{C}^a_{A_1
A_k})^2$ and $(\mathcal{C}^a_{A_1 A_{k+1}})^2-T$, the second
inequality is by iterative use of Eq.~(\ref{gqpoly}) on
$(\mathcal{C}^a_{A_1 A_2})^2+\cdots+(\mathcal{C}^a_{A_1 A_k})^2$,
and the last inequality is by Eq.~(\ref{gTEoA}).

For a $n$-qubit mixed state  $\rho_{A_1(A_2\cdots A_n)}$, let
$\rho_{A_1(A_2\cdots A_n)}=\sum_j p_j \ket{\psi_j}_{A_1(A_2\cdots
A_n)}\bra{\psi_j}$ be an optimal decomposition for TEoA such that
${\mathcal T}^a_{q}\left(\rho_{A_1(A_2\cdots A_n)}\right)=\sum_j p_j
{\mathcal T}_{q}\left(\ket{\psi_j}_{A_1(A_2\cdots A_n)}\right)$.
Because each $\ket{\psi_j}_{A_1(A_2\cdots A_n)}$ in the
decomposition is an $n$-qubit pure state, we have
\begin{widetext}
\begin{align}
{\mathcal T}^a_{q}\left(\rho_{A_1(A_2\cdots A_n)}\right)=&\sum_j p_j
{\mathcal T}^a_{q}\left(\ket{\psi_j}_{A_1(A_2\cdots
A_n)}\right)\nonumber\\
\leq&\sum_j p_j\left({\mathcal
T}^a_{q}\left(\rho^j_{A_1A_2}\right)+\cdots +{\mathcal
T}^a_{q}\left(\rho^j_{A_1A_n}\right) \right)\nonumber\\
=&\sum_j p_j{\mathcal T}^a_{q}\left(\rho^j_{A_1A_2}\right)+\cdots
+\sum_j p_j{\mathcal
T}^a_{q}\left(\rho^j_{A_1A_n}\right) \nonumber\\
\leq&{\mathcal T}^a_{q}\left(\rho_{A_1A_2}\right)+\cdots +{\mathcal
T}^a_{q}\left(\rho_{A_1A_n}\right), \label{Tpolymixed}
\end{align}
\end{widetext}
where $\rho^j_{A_1A_i}$ is the reduced density matrix of
$\ket{\psi_j}_{A_1(A_2\cdots A_n)}$ onto subsystem $A_1A_i$ for each
$i=2,\cdots,n$ and the last inequality is by definition of TEoA for
each $\rho_{A_1A_i}$.
\end{proof}

Although Theorem~\ref{Thm: poly} provides the polygamy inequality of multi-qubit
entanglement in terms of TEoA for $1\leq q \leq 2$ or $3\leq q \leq 4$, it is also
clear that Eq.~(\ref{Tpoly}) is also true for $q$ slightly larger than 4 or
less than 1 due to its continuity with respect to $q$.


\section{Conclusion}
\label{Conclusion}

Using Tsallis-$q$ entropy, we have established a class of bipartite
entanglement measures, Tsallis-$q$ entanglement, and provided its
analytic formula in two-qubit systems for $1\leq q\leq 4$. Based on
the functional relation between concurrence and Tsallis-$q$
entanglement, we have shown that the monogamy of multi-qubit
entanglement can be mathematically characterized in terms of
Tsallis-$q$ entanglement for $2\leq q\leq 3$. We have also provided
a polygamy inequality of multi-qubit entanglement in terms of TEoA
for $1\leq q\leq 2$ and $3\leq q\leq 4$.

The class of monogamy and polygamy inequalities of multi-qubit
entanglement we provided here consists of infinitely many
inequalities parameterized by $q$. We believe that our result will
provide useful tools and strong candidates for general monogamy and
polygamy relations of entanglement in multipartite
higher-dimensional quantum systems, which is one of the most
important and necessary topics in the study of multipartite quantum
entanglement.

\section*{Acknowledgments}
This work was supported by {\it i}CORE, MITACS and USARO.


\begin{thebibliography}{1}

\bibitem{tele}
C. H. Bennett, G. Brassard, C. Crepeau, R. Jozsa, A. Peres and W. K.
Wootters, Phys. Rev. Lett. {\bf 70}, 1895 (1993).

\bibitem{noclon}
W. K. Wootters and W. H. Zurek,
Nature {\bf 299}, 802 (1982).

\bibitem{T04}
B. M. Terhal, IBM J. Research and Development 48, 71 (2004).

\bibitem{m}
L. Masanes,
Phys. Rev. Lett. {\bf 102}, 140501 (2009).

\bibitem{ww}
W. K. Wootters,
Phys. Rev. Lett. {\bf 80}, 2245 (1998).

\bibitem{ckw}
V. Coffman, J. Kundu and W. K. Wootters, Phys. Rev. A {\bf 61},
052306 (2000).

\bibitem{ov}
T. Osborne and F. Verstraete, Phys. Rev. Lett. {\bf 96}, 220503
(2006).

\bibitem{bdsw}
C. H. Bennett, D. P. DiVincenzo, J. A. Smolin and W. K. Wootters,
Phys. Rev. A {\bf 54}, 3824 (1996).

\bibitem{renyi}
A. R\'enyi,  {\em Proceedings of the Fourth Berkeley
Symposium on Mathematics, Statistics and Probability}
(University of California Press, Berkeley, 1960) {\bf 1}, p. 547-561 .

\bibitem{horo}
R. Horodecki, P. Horodecki and M. Horodecki,
Phys. Lett. A {\bf 210}, 377 (1996).

\bibitem{tsallis}
C. Tsallis, J. Stat. Phys. {\bf 52}, 479 (1988).

\bibitem{lv}
P. T. Landsberg and V. Vedral, Phys. Lett. A {\bf 247}, 211 (1998).

\bibitem{ks2}
J. S. Kim and B. C. Sanders, arXiv.org:0911.5180 (2009).

\bibitem{cohen}
O. Cohen, Phys. Rev. Lett. {\bf 80}, 2493 (1998).

\bibitem{suffi}
Due to the existence of the decomposition satisfying
Eqs~(\ref{Copt}) and (\ref{Cphii}), we have
\begin{align*}
g_{q}\left(\mathcal{C}(\rho_{AB}) \right)=&g_{q}\left(\sum_i
p_i\mathcal{C}(\ket{\phi_i}_{AB})\right)\nonumber\\
=&\sum_i
p_ig_{q}\left(\mathcal{C}(\ket{\phi_i}_{AB})\right)\nonumber\\
=&\sum_i
p_i{\mathcal T}_{q}(\ket{\phi_i}_{AB})\nonumber\\
\geq&{\mathcal T}_{q}\left(\rho_{AB}\right).
\end{align*}

Conversely, the existence of the optimal decomposition of
$\rho_{AB}=\sum_j q_j\ket{\mu_j}_{AB}\bra{\mu_j}$ for Tsallis-$q$
entanglement leads us to
\begin{align*}
{\mathcal T}_{q}\left(\rho_{AB}\right)=&\sum_j q_j{\mathcal
T}_{q}\left(\ket{\mu_j}_{AB}\right)\nonumber\\
=&\sum_j q_jg_{q}\left({\mathcal C}(\ket{\mu_j}_{AB})\right)\nonumber\\
\geq&g_{q}\left(\sum_j q_j{\mathcal C}(\ket{\mu_j}_{AB})\right)\nonumber\\
\geq&g_{q}\left({\mathcal C}(\rho_{AB})\right),\nonumber\\
\end{align*}
where the first and second inequalities are due to the convexity and
monotonicity of $g_q(x)$.

\bibitem{lve}
T. Laustsen, F. Verstraete and S. J. van Enk, Quantum Inf. Comput.
{\bf 3}, 64 (2003).

\bibitem{gbs}
G. Gour, S. Bandyopadhay and B. C. Sanders, J. Math. Phys. {\bf
48}, 012108 (2007).

\bibitem{bgk}
F. Buscemi, G. Gour and J. S. Kim,
Phys. Rev. A {\bf 80}, 012324 (2009).

\end{thebibliography}
\end{document}